\documentclass{cccg25}
\usepackage{geometry}
\usepackage[T1]{fontenc}
\usepackage{graphicx}
\usepackage{multicol}
\usepackage{amssymb,amsmath }

\usepackage{xcolor}
\usepackage{subcaption}
\usepackage{changepage}

\newcommand\sasa[1]{\textcolor{black}{#1}}

\newcommand\ro[1]{\textcolor{black}{#1}}

\newcommand\roA[1]{\textcolor{black}{#1}}
\newcommand\roB[1]{\textcolor{black}{#1}}

\newcommand\roD[1]{\textcolor{black}{#1}}
\newcommand\roE[1]{\textcolor{black}{#1}}
\newcommand\roF[1]{\textcolor{black}{#1}}
\begin{document}

\title{Generalized k-Cell Decomposition for Visibility Planning in Polygons}
%
%
%
%



\title{Generalized k-Cell Decomposition for Visibility Planning in Polygons}

\author{%
  \begin{tabular}{cc}
    \begin{tabular}[t]{c} 
      1\textsuperscript{st} Yeganeh Bahoo \textit{Computer Science} \\
      Toronto Metropolitan University \\
      Toronto, Canada \\
      0000-0001-5349-494
    \end{tabular}
    &
    \begin{tabular}[t]{c}
      2\textsuperscript{nd} Sajad Sacedi \textit{Computer Science} \\
      University College London \\
      London, UK \\
      0000-0002-6385-6127
    \end{tabular}
    \\[3ex] 
    \multicolumn{2}{c}{
      \begin{tabular}{c}
        3\textsuperscript{rd} Roni Sherman \textit{Computer Science} \\
        Toronto Metropolitan University \\
        Toronto, Canada \\
        0009-0004-0542-3480
      \end{tabular}
    }
  \end{tabular}
}
\maketitle

%

%


\begin{abstract}

This paper introduces a novel $k$-cell decomposition method for pursuit-evasion problems in polygonal environments, where a searcher is equipped with a $k$-modem: a device capable of seeing through up to $k$ walls. The proposed decomposition ensures that as the searcher moves within a cell, the structure of unseen regions (shadows) remains unchanged, thereby \ro{preventing any geometric events between or on invisible regions, that is,} preventing the appearance, disappearance, merge, or split of shadow regions. The method extends existing work on $0$- and $2$-visibility by incorporating m-visibility polygons for all even $0 \le m \le k$, constructing partition lines that enable robust environment division. \roF{The correctness of the decomposition is proved} via three theorems. 
The decomposition enables reliable path planning for intruder detection in simulated environments and opens new avenues for visibility-based robotic surveillance. \roF{The difficulty in constructing the cells of the decomposition consists in computing the $k$-visibility polygon from each vertex and finding the intersection points of the partition lines to create the cells.}

\end{abstract}

\section{Introduction}
Pursuit Evasion~\cite{ChungHollingerIsler11} is a famous problem in computational geometry, especially for robotics and game theory. Pursuit evasion involves one or more pursuers attempting to locate or catch one or more evaders in some environment. There are four main parameters which describe pursuit evasion: the type of environment, the number of pursuers and evaders, the speed of the pursuers and evaders, and the field of vision of the pursuers and evaders. In Pursuit Evasion games, there is also always a strategy involved \roF{in how the pursuer attempts to capture the evader and the evader tries to avoid capture. }

The contribution of this work is the development of a general $k$-cell decomposition method that guarantees stability of shadow regions during intra-cell movement. We provide a formal proof that no appear, disappear, split, or merge events occur within any cell, extending previous work limited to specific $k$-values. 

%

\roD{The main difference between this paper and~\cite{Bahoo2013} is that this paper generalizes the idea for $2$-visibility cell decomposition of the polygons to $k$-visibility with a proof of correctness that extends the two proofs in~\cite{Bahoo2013}. It should be noted that this generalization was not trivial. The proof of correctness for vertex-shadow and edge-shadow has been modified to work for $k$-visibility, with the proof of edge-shadow being condensed.}
\sasa{
The rest of the paper is organized as follows: 
Sec.~\ref{sec:bck} presents the background material. 
Sec.~\ref{sec:method} describes the method. 
Sec.~\ref{sec:conclusions} concludes the paper.
}


\section{Background}\label{sec:bck}
Consider a simple polygon $P$ in $2D$ with some evaders and a pursuer moving inside of it. \roA{The entire map is known. The pursuers and evader have
infinite speed. The evader knows the pursuer’s location and all of its movements. The pursuers and
evaders move continuously within the polygon.} As the pursuer moves, parts of the polygon might become invisible to it. 
Two points $p$ and $q$ are said to be visible when the line segment $pq$ does not intersect the polygon. This is shown in Figure~\ref{fig:visibility}. In particular, we are interested in $k$-visibility, where two points $p$ and $q$ are said to be $k$-visible when the segment $pq$ intersects the polygon at most $k$ times. An example of this is shown in Figure~\ref{fig:k-visibility}, where two walls ($k=2$) separate $p$ and $q$.

Let $p$ be the initial position of the pursuer, an arbitrary point within the polygon $P$. Each maximal connected set of points within the polygon $P$ that is invisible to $p$ forms what is called a shadow of $p$. The shadows of $p$ are sub-polygons of $P$, denoted by $S_i(p)$. As the pursuer (searcher) moves continuously in $P$, four geometric events may occur for its shadow: merge, split, appear, or  disappear~\cite{yu2011shadow} \roE{(See Figure~\ref{fig:four_events})}.  


The challenge is to develop a cell decomposition such that when a searcher moves within a cell, none of the four events \roF{occur}\ro{~\cite{Guiba1997}.
To calculate the $k$ visibility polygon, there are three algorithms available. Martins et al.~\cite{Martins2009} presented a $O(n^2)$ algorithm. Bahoo et al.~\cite{BahooYeganeh2020CtkR} improved upon this with a $O(nlog(n))$ algorithm. An even faster algorithm by Bahoo et al.~\cite{BahooYeganeh2020CtkR} (for values of $k < log(n)$) runs in $O(kn)$}.






\begin{figure}[h]
\centering
\begin{subfigure}[b]{.35\linewidth}
\includegraphics[width=\linewidth]{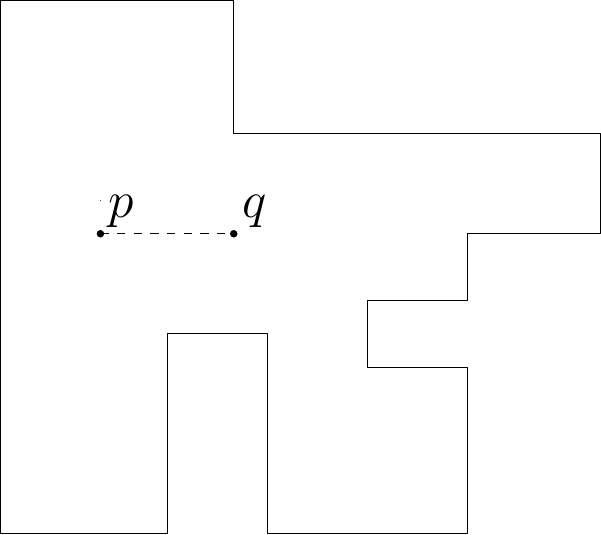}
\caption{0-visibility}\label{fig:visibility}
\end{subfigure}
\begin{subfigure}[b]
{.35\linewidth}
\includegraphics[width=\linewidth]{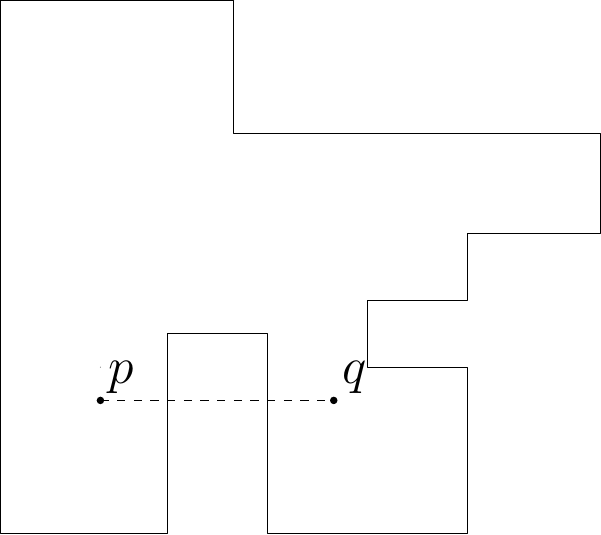}
\caption{2-visibility}\label{fig:k-visibility}
\end{subfigure}

\caption{Visibility between two points $p$ and $q$}
\label{fig:four_events}
\end{figure}

\begin{figure}
\centering
\begin{subfigure}[b]{.25\linewidth}
\includegraphics[width=\linewidth]{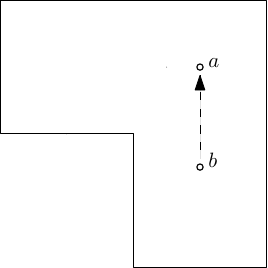}
\caption{Disappear}\label{fig:disappear}
\end{subfigure}
\begin{subfigure}[b]{.25\linewidth}
\includegraphics[width=\linewidth]{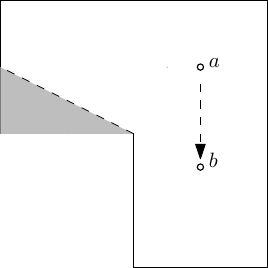}
\caption{Appear}\label{fig:appear}
\end{subfigure}

\begin{subfigure}[b]{.25\linewidth}
\includegraphics[width=\linewidth]{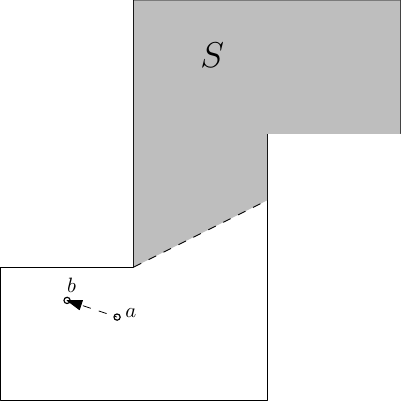}
\caption{Merge}\label{fig:merge}
\end{subfigure}
\begin{subfigure}[b]{.25\linewidth}
\includegraphics[width=\linewidth]{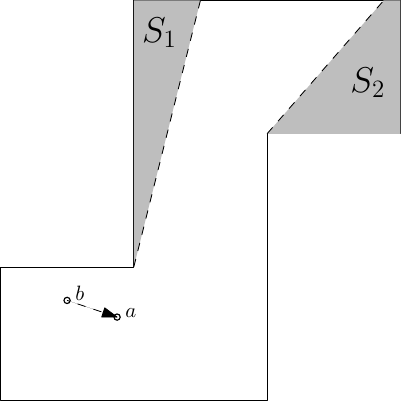}
\caption{Split}\label{fig:split}
\end{subfigure}
\caption{The four geometric events that may occur to the pursuer's shadow as the pursuer moves within the polygon}
\label{fig:four_events}
\end{figure}
\section{Proposed Method}\label{sec:method}

First the cell decomposition is presented, and \ro{then a proof of correctness \roF{is} proposed}.
The proof consists of three theorems.

\subsection{The K-Cell Decomposition}
 The $k$-cell decomposition is created by computing all even-valued visibility polygons, from 0-visibility to $k$-visibility, \roB{at} each vertex of $P$. We then use the lines that define each polygon as the partition lines. In other words, we calculate the following visibility polygons:
\begin{itemize}
    \item Lines of the $k$-visibility polygon of each vertex. 
    \item Lines of the $k-2$-visibility polygons of each vertex.
    \item Lines of the $k-4$-visibility polygon of each vertex. 
    ...
    \item Lines of the $0$-visibility polygon of each vertex.
\end{itemize}


An example cell decomposition is presented in   Figure~\ref{fig:cell_decomp_full_path_planning} : 

\begin{figure}[h]
    \centering
    \includegraphics[width=.55\columnwidth]{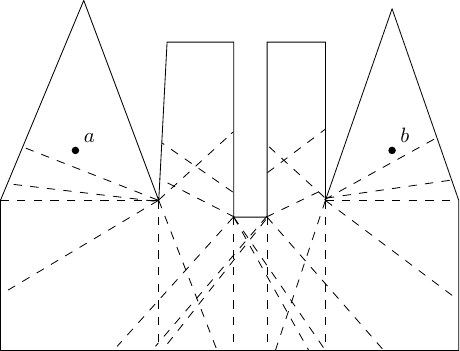}
    \caption[A cell decomposition with lines drawn from every vertex]{The cell decomposition with all decomposition lines drawn from every vertex, \ro{for $k=2$} }
    \label{fig:cell_decomp_full_path_planning}
\end{figure}


\roD{The intuition behind this is that a vertex always participates in an event, therefore considering all combinations of vertices at all values of $k$ gives an upper bound solution.}

\subsection{Proof of Correctness of the K-Cell Decomposition}
 \ro{The proof consists of two theorems for two different types of shadow, and one more theorem which summarizes the two. 
 }
\begin{definition}
 A type 1 shadow (vertex shadow) is defined as a shadow that contains a vertex \cite{Bahoo2013}. See Figure~\ref{fig:vertex_shadow}.
\end{definition}

\begin{definition}
A type 2 shadow (edge shadow) is a shadow that does not include a vertex \cite{Bahoo2013}. This shadow occurs only between edges. See Figure~\ref{fig:edge_shadow}. \roB{Note that an edge shadow cannot occur for the case of $k=0$.}
\end{definition}
\begin{definition}
     A partition line is a line of the cell decomposition (which partitions the space into cells)
\end{definition}
\begin{definition}
    A vertex $a$ that is critical to some point $b$ is a vertex with both of its edges to one side of the line $ab$.
\end{definition}

 \begin{figure}[h]
    \centering
    \includegraphics[width=0.35\columnwidth]{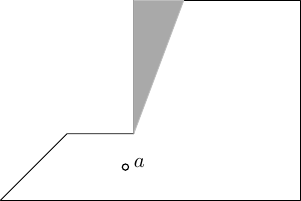}
    \caption{The shadow of type 1 - a shadow that includes a vertex}
    \label{fig:vertex_shadow}
\end{figure}
\begin{theorem}
    (Invariance of type 1 Shadow)
When an agent moves in a cell $C$, the combinatorial representation of the shadow regions of type 1 remains unchanged.
\end{theorem}
\begin{proof}
    Consider a point $p$ that is inside a cell 
    $C$, a point $q$ that is also inside the cell $C$, and a vertex $v$ of the polygon that is invisible to $p$, i.e., in the shadow of $p$ (shadow of type 1) (see Figure~\ref{fig:p2_s1_1}). Assume for contradiction that $v$ is visible to $q$. When the pursuer moves from $p$ to $q$, the pursuer crosses a line of the cell decomposition (See Figure~\ref{fig:p2_s1_2}), so $p$ and $q$ are not in the same cell. 
\end{proof}

 \begin{figure}[h]
    \centering
    \includegraphics[width=0.4\columnwidth]{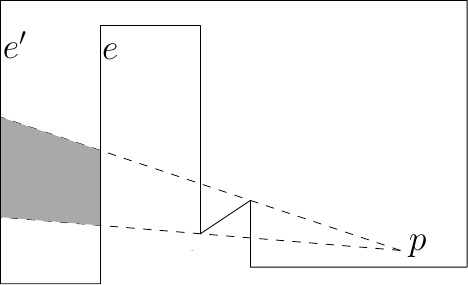}
    \caption{The shadow of type 2 (edge shadow) - a shadow that occurs between two edges}
    \label{fig:edge_shadow}
\end{figure}
\begin{figure}[h]
    \centering

    \begin{subfigure}[t]{.40\linewidth}
\includegraphics[width=\linewidth]{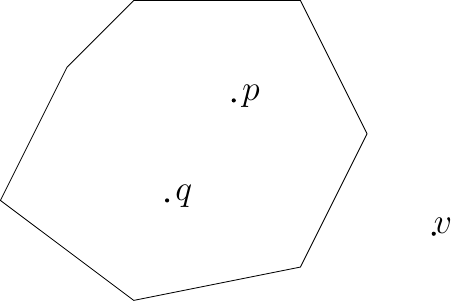}
\caption{A point $p$ and a point $q$ inside a cell $C$, and a vertex of a shadow of type 1}\label{fig:p2_s1_1}
\end{subfigure} \hfill
\begin{subfigure}[t]{.40\linewidth}
\includegraphics[width=\linewidth]{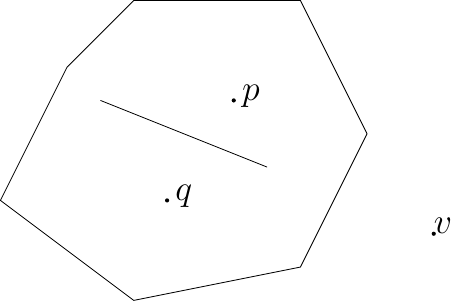}
\caption{A partition line between the vertex $p$ and vertex $q$}\label{fig:p2_s1_2}
\end{subfigure}

\caption{Figures for Theorem 1}
    \label{fig:theorem1}



\label{fig:}
\end{figure}







Before proceeding to the second theorem, we first introduce the necessary definitions and establish several supporting lemmas.

\roE{Consider two points, $p$ and $q$, located within the same cell $C$ of the decomposition, which denote the  position of an agent at two different times as it moves within the polygon (See Figure~\ref{fig:shadow_of_p}). There are two edges of the polygon $e$ and $e'$, with $e'$ being the next wall after $e$ on the ray $pp'$. From the point $p$, there is a shadow of type 2 between these two edges. Two rays, $pp'$ and $pp''$, define this shadow ($p'$ and $p''$ being points on $e'$).}
\roE{Consider a ray emanating from $q$ which rotates counterclockwise around $q$. Consider the first time this ray hits both $e$ and $e'$. The intersection of this ray with $e$ is $q'$. The last time this ray intersects both $e$ and $e'$ it intersects $e'$ at $q''$.}

We assign a coordinate system in which the point $p$ is at the origin $(0,0)$ and the point $q$ is on the $y$-axis. 
Consider intersection points which further describe the geometric layout: Let $t$ be the point that is the intersection between the line $qq''$ and $pp'$. Let $s$ be the intersection between the line $qq''$ and $pp''$. This intersection points help with definitions in the proof. \roB{There exists a vertex, $m$, on the segment $pp'$ that is critical to $p$.} Call one of the edges of vertex $m$, the edge of $m$ that makes the smallest angle with the $x$-axis (as defined by the coordinate system) $e_m$, as in Figure~\ref{fig:vertex_m}.


\begin{figure}[h]
    \centering
    \includegraphics[width=0.55\columnwidth]{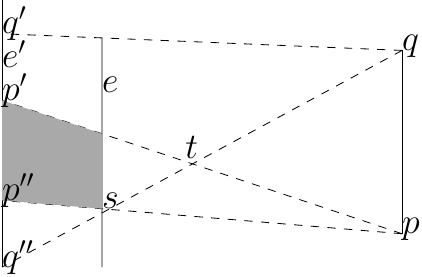}
    \caption{The shadow of $p$. The intersection points $t$ and $s$ with $qq''$}
    \label{fig:shadow_of_p}
\end{figure}

\begin{figure}[h]
    \centering
    \includegraphics[width=0.55\columnwidth]{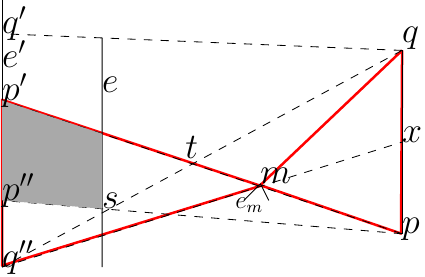}
    \caption{The existence of a vertex, $m$, on $pp'$ which `starts' the shadow of $p$}
    \label{fig:vertex_m}
\end{figure}

\begin{lemma} There exists a vertex, $m$, on the segment $pp'$, with both edges lying below $pp'$. The vertex $m$ is critical to $p$.
\end{lemma}

\begin{proof}

\roE{Since the ray $pp'$ is responsible for forming a part of the shadow of $p$, it means that somewhere along this segment, there is a vertex which obstructs the visibility from the point $p$. This vertex has both of its adjacent edges below segment $pp'$.}

\end{proof}

The next 
Lemma will be established through proof by contradiction, as follows: 
we first consider that $p$ has a shadow of Type 2 between edges $e$ and $e'$. Moreover, \roE{for appear/disappear event,} we suppose that the area between edges $e$ and $e'$ – the interior of the triangle $qq'q''$ that is between $e$ and $e'$ – is visible to point $q$. 



    \begin{figure}
\centering
\begin{subfigure}[b]{.40\linewidth}
\includegraphics[width=\linewidth]{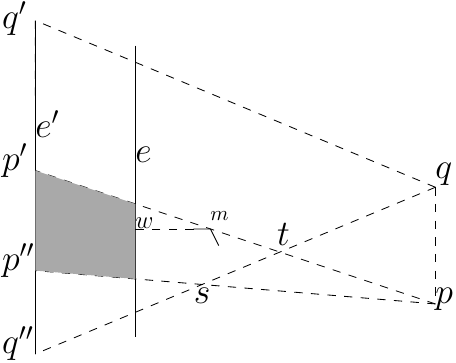}
\caption{Case 1}\label{fig:case1}
\end{subfigure}
\begin{subfigure}[b]{.40\linewidth}
\includegraphics[width=\linewidth]{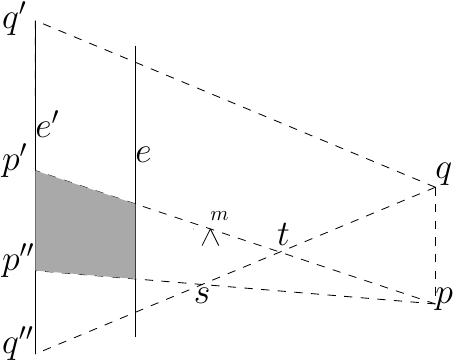}
\caption{Case 2}\label{fig:case2}
\end{subfigure}

\begin{subfigure}[b]{.40\linewidth}
\includegraphics[width=\linewidth]{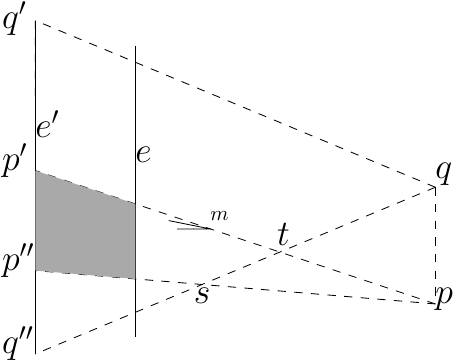}
\caption{Case 3}\label{fig:case6}
\end{subfigure}
\begin{subfigure}[b]{.40\linewidth}
\includegraphics[width=\linewidth]{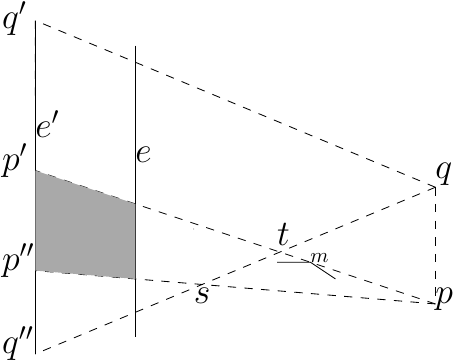}
\caption{Case 4}\label{fig:case3}
\end{subfigure}

\begin{subfigure}[b]{.40\linewidth}
\includegraphics[width=\linewidth]{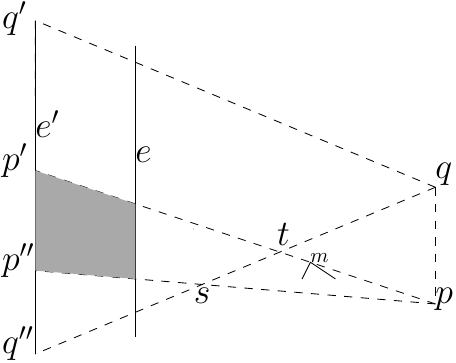}
\caption{Case 5}\label{fig:case4}
\end{subfigure}
\begin{subfigure}[b]{.40\linewidth}
\includegraphics[width=\linewidth]{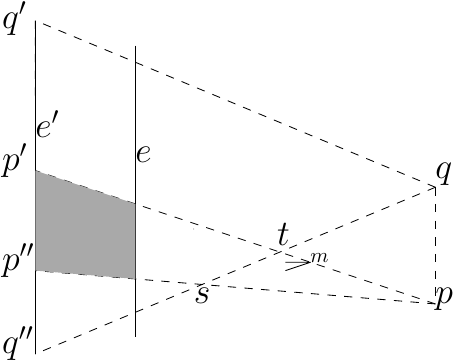}
\caption{Case 6}\label{fig:case5}
\end{subfigure}
\caption{The Six Cases for Theorem 2}
\label{fig:six_cases}
\end{figure}

\begin{figure}[h]
    \centering
    \includegraphics[width=0.30\columnwidth]{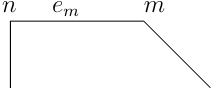}
    \caption{The other vertex of edge $e_m$, `$n$'}
    \label{fig:other_vertex_of_m}
\end{figure}

\begin{lemma}
\label{lemma2}
\roB{If a $k$-modem that is in a cell $C$ of the decomposition and with a shadow of type 2 moves continuously inside the cell, appear and disappear events of  the shadow regions of type 2 may not happen.  }
\end{lemma}
\begin{proof}
    
\roE{The proof is by contradiction.} Assume\roE{, for contradiction,} that there is no vertex in $\triangle{mpq}$ or in $\triangle{mp'q''}$ (\roE{See the two red triangles in Figure.~\ref{fig:vertex_m}}) that is critical to $m$ and forms a partition line with $m$. 
This means that as $pp'$ is rotated around $m$ towards $q$, at every vertex that is critical to $m$, including the other vertex of $e_m$ (that is not $m$, see Figure~\ref{fig:other_vertex_of_m}), there are $k$ or less walls on both the right and left sides of the rotated line, \roE{regardless of which of the six cases we are in (See Figure~\ref{fig:six_cases}}). Thus, there are always less than $k$ walls below $pp'$ (unrotated), so $pp'$ is not the boundary of the shadow of $p$, a contradiction.
\roE{Note that case 3 is a subcase of 1 and case 6 is a subcase of 4.}

\end{proof}

\begin{lemma}
    The $k$-visibility polygon drawn from each vertex ensures that no merge or split occurs \roF{when} going from one cell to another.
\end{lemma}
\begin{proof}
    If suddenly one shadow splits into two \roF{while} going from $p$ to $q$, it means that there is a new critical vertex $c$ that has appeared \roF{to} $q$ which was not visible \roF{at} $p$. As such, the partition line created by the $k$-visibility polygon of $c$ must intersect the segment $pq$. So $p$ and $q$ may not lie in the same cell of the decomposition, a contradiction.
\end{proof}



This concludes our final theorem, which is as follows: 

\begin{theorem}
    The $k$-cell decomposition guaranties that for a shadow of type 2 no geometric events (i.e., appear, disappear, merge, and split) will occur while a pursuer moves continuously in a cell.  
\end{theorem}

\begin{proof}
    This follows as a consequence of Lemmas 1, 2, and 3. \roD{Note that there are only two possible shadows - vertex shadow and edge shadow.}
\end{proof}

\begin{theorem}
    The proposed cell decomposition is complete. In other words, \roF{while} moving within a cell, no geometric events may happen \roF{to} the shadows. 
\end{theorem}

\begin{proof}
    This is proved by Theorem 1 and 2.
\end{proof}

\subsection{Computational Complexity}

\roE{The worst case complexity is when all vertices are critical for each other. As such, for a vertex $v$, there may exist $O(n)$ partition lines for each $i$ in ${0,2,4\dots,k}$. So, each vertex, there are $O(nk)$ partition lines. Overall there are $n$ vertices. As such the total number of partition lines is $O(kn^2)$. For a particular vertex, each partition line can intersect all the rest of the partition lines except the one emanating from the same vertex. Consequently, each partition line may intersect $O(kn^2-kn)$ other partition lines. Each vertex has $O(nk)$ partition lines. So, all the partition lines emanating from a particular vertex intersect $O(nk(kn^2-kn))$. We have overall $n$ vertices. As such the number of intersections (vertices of the cell decomposition) is $O(k^2n^4)$.}

\subsection{Applications}

\roE{The primary purpose of the cell decomposition  is to compute a path for one or multiple pursuers that guarantees detection of an intruder, and the same approach presented by Guibas et al. \cite{Guiba1997} works on the cell decomposition presented in this paper.}
We extend this concept for $k$-modem\roB{s} by building this graph on the $k$-cell decomposition. Using this, we propose \roF{an} algorithm that \roF{computes} a path for detecting intruders in a given polygon using a $k$-modem as the pursuer, if there exist\roF{s} any. 

\section{Conclusion} \label{sec:conclusions}  

This paper studied the pursuit-evasion problem under the $k$-visibility model, 
\roA{generalizing} existing work on $k=0$ and $k=2$ visibility-based decompositions. The focus of this research was to develop a cell decomposition that allows a pursuer equipped with a $k$-modem to navigate within a polygonal environment to detect an intruder. 
To do so, the environment is split into cells, ensuring that the combinatorial representation of the shadow regions remains unchanged as an agent moves within a cell. This means that none of the key visibility events, that is, merge, split, appear, or disappear occur as the pursuer moves within a cell of the decomposition. Future work could include reducing the number of lines \roF{in} the decomposition \ro{or path planning between two given points to avoid specific geometric event that might occur along the path}. 

\newcommand\blfootnote[1]{%
  \begingroup
  \renewcommand\thefootnote{}\footnote{#1}%
  \addtocounter{footnote}{-1}%
  \endgroup
}
\blfootnote{This work was supported by the Natural Sciences and Engineering Research Council of Canada (NSERC).}

%
%
%
%
\bibliographystyle{abbrv}

\bibliography{bibliography}

@inproceedings{Bahoo2013,
author = {Bahoo, Yeganeh and Mohades, Ali and Eskandari, Marzieh and Sorouri, Mahsa},
booktitle={29th European Workshop on Computational Geometry},
year = {2013},
pages = {201--204},
title = {{2-Modem Pursuit-Evasion Problem}}
}

@article{yu2011shadow,
  title={Shadow information spaces: Combinatorial filters for tracking targets},
  author={Yu, Jingjin and LaValle, Steven M},
  journal={IEEE Transactions on Robotics},
  volume={28},
  number={2},
  pages={440--456},
  year={2011},
  publisher={IEEE}
}

@phdthesis{Martins2009,
author = {Ana Mafalda
de Oliveira Martins},
year = {2009},
title = {Geometric optimization on visibility problems: Metaheuristic and exact solutions},
 school = {Universidade de Aveiro (Portugal)}
}

@InProceedings{Guiba1997,
author = {Guibas, Leonidas J.
and Latombe, Jean-Claude
and Lavalle, Steven M.
and Lin, David
and Motwani, Rajeev}, 
editor = {Dehne, Frank
and Rau-Chaplin, Andrew
and Sack, J{\"o}rg-R{\"u}diger
and Tamassia, Roberto},
title = {Visibility-based pursuit-evasion in a polygonal environment},
booktitle = {Algorithms and Data Structures},
year = {1997},
publisher = {Springer Berlin Heidelberg},
address = {Berlin, Heidelberg},
pages = {17--30}
}

@article{ChungHollingerIsler11, 
author = {Chung, Timothy and Hollinger, Geoffrey and Isler, Volkan},
year = {2011},
month = {11},
pages = {},
title = {Search and Pursuit-evasion in Mobile Robotics},
volume = {31},
journal = {Auton. Robots},
doi = {10.1007/s10514-011-9241-4}
}

@article{BahooYeganeh2020CtkR,
author = {Bahoo, Yeganeh and Bose, Prosenjit and Durocher, Stephane and Shermer, Thomas C.},
address = {New York},
copyright = {Springer Science+Business Media, LLC, part of Springer Nature 2020},
issn = {1432-4350},
journal = {Theory of computing systems},
language = {eng},
number = {7},
pages = {1292-1306},
publisher = {Springer US},
title = {Computing the k-Visibility Region of a Point in a Polygon},
volume = {64},
year = {2020},
}
\end{document}